\documentclass[10pt,aps,letterpaper,pra,twocolumn]{revtex4-1}

\usepackage{graphicx}
\usepackage[all]{xy}
\usepackage[linkcolor=red]{hyperref}
\usepackage{amsmath}
\usepackage{amsthm}
\usepackage{amssymb}

\allowdisplaybreaks

\newtheorem{thm}{Theorem}
\newtheorem{lem}[thm]{Lemma}

\newcommand{\ket}[1]{{|}{#1}{\rangle}}

\newcommand{\pset}[2]{\mathcal{P}^{#1}_{#2}}

\begin{document}

\title{Quantum Oracles in Constant Depth with Measurement-Based\texorpdfstring{\\}{}
  Quantum Computation}
\date{\today}
\author{Beno\^{\i}t Valiron}
\affiliation{PPS, UMR 7126, Universit\'e Paris Diderot,
  Sorbonne Paris Cit\'{e}, F-75205 Paris, France}

\begin{abstract}
  This paper shows that, in measurement-based quantum computation, it
  is possible to write any quantum oracle implementing a classical
  function in constant depth. The result is shown through the
  equivalence between MBQC and the circuit model where arbitrary
  rotations along $Z$ axis and unbounded fan-outs are elementary
  operations. A corollary of this result is that disjunction can be
  implemented exactly in constant-depth, answering an open question of
  H{\o}yer and \v{S}palek.
\end{abstract}


\maketitle

Proposed by Raussendorf and Briegel~\cite{RB01}, the measure\-ment-based
quantum computational model is radically different from the circuit
model. In the latter, the computation is performed on a set of quantum
bit registers by successive applications of quantum gates. On the
contrary, in the former the computation proceeds by adaptative
one-qubit measurements performed on a {\em cluster state}, that is, a particular
entangled multi-qubit state. The computation is encoded in the graph
of entanglement, in the choice of basis for the measurements, and in
their dependency graph.

In the measurement-based quantum computational model, the {\em depth} of the
computation is the longest path in the dependency graph. Browne,
Kashefi and Perdrix~\cite{browne10} show that this model is
computationally equivalent to the circuit model where arbitrary
rotations $R(\theta)$ around the $Z$ axis
\[
R(\theta) = 
\begin{pmatrix}
  1&0\\
  0&e^{\frac{i\theta}2}
\end{pmatrix},
\]
unbounded fan-outs and parity gates are
taken as elementary gates. In particular, the depth-complexity of an
algorithm is the same in both models, provided that {\em classical}
unbounded parity gates are free.

Because of decoherence, the depth of an algorithm is a crucial
limitation for quantum computation: in general, we want quantum
algorithms to be as parallel as possible. Measurement-based quantum
computation is a natural parallel computational paradigm and various
works investigate its capabilities in terms of depth of
computations~\cite{broadbent09, hoyer05, browne10}. In particular, if
one considers {\em approximations} and not exact descriptions, several
algorithms can be implemented in constant depth~\cite{hoyer05}.

\bigskip
This paper presents a novel result with respect to exact descriptions:
the fact that quantum oracles of the form
$\ket{x}\ket{y}\mapsto\ket{x}\ket{y\oplus f(x)}$ can be {\em exactly}
encoded in constant depth in measurement-based quantum computation. 

As shown in Section~\ref{sec:naive}, it is clear that quantum circuits
can easily do it with a suitable choice of elementary gates, provided
that we allow circuits to have a width exponential on the size of the
input. However, the fact measurement-based quantum computation can
also do it has not been shown so far.

In order to prove this result, we use the equivalent representation in
term of quantum circuits presented by Browne and al.~\cite{browne10},
and we generalize the decomposition of the Toffoli gate given by
Selinger~\cite{selinger}. As a side effect, we also answer an open
question of H{\o}yer and \v{S}palek~\cite{hoyer05}: there {\em is} a
constant-depth exact circuit for the disjunction boolean operator.

\section{Naive parallel implementation of quantum oracles}
\label{sec:naive}

Consider a boolean function $f$ on $n$ inputs. This boolean function
can always be written as
\begin{equation}
\label{eq:f}
(x_1,\ldots,x_n) ~~\longmapsto~~ \bigoplus_{i=1}^N\bigwedge_{k\in K_i}\!\!x_k
\end{equation}
where $N$ is some natural number and the $K_i$'s some subsets of
indices.

\begin{figure}
\begin{center}
  \includegraphics[width=3in]{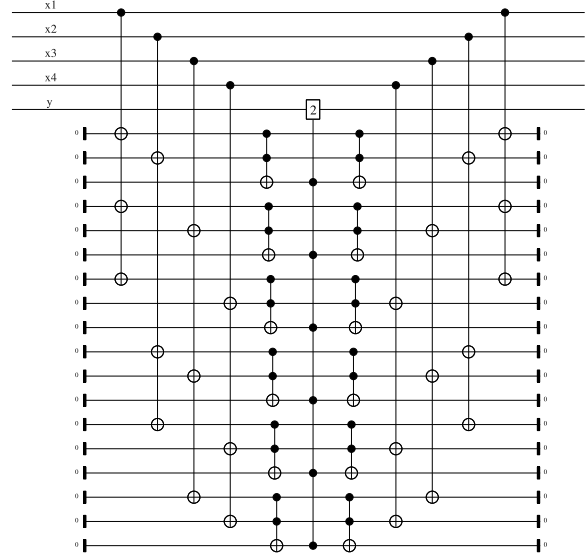}
\end{center}
\caption{Depth-5 oracle with multi-controlled CNOTs, fanouts and
  parity gates.}
\label{fig:1}
\end{figure}

Provided that multi-controlled NOT-gates, unbounded fanouts and
unbounded parity gates are available as elementary gates, this
function can trivially be implemented as a quantum oracle of the form
\[
\ket{x_1,\ldots,x_n}\ket{y}  ~~\longmapsto~~
\ket{x_1,\ldots,x_n}\ket{y\oplus{}f(x_1,\ldots x_n)},
\]
in constant depth, as follows:
\begin{enumerate}
\item Allocate one block of ancillas for each $K_i$. The
  $i$-th block is of the size of $K_i$, plus one.
\item For each $i$, copy $y$ and $\{x_k\,|\,k\in K_i\}$ to the
  corresponding block. This can be done in one step, with
  fanouts in parallel.
\item perform the conjunctions on each block, using multi-controlled
  NOT-gates. Again, this is one step.
\item Do the final XOR on the $y$ gate using a parity gate (one step).
\item Undo the ancillas: multi-controlled
  NOT-gates, then fanouts (two steps).
\item Desallocate the ancillas
\end{enumerate}
In total, not counting allocation and desallocation, the depth of the
circuit is $5$. As an example, the function 
\begin{multline}
\label{eq:f4}
f(x_1,x_2,x_3,x_4) =
(x_1\wedge x_2) \oplus (x_1\wedge x_3)\\ \oplus
(x_1\wedge x_4) \oplus (x_2\wedge x_3) \oplus (x_2\wedge x_4) \oplus
(x_3\wedge x_4)
\end{multline}
can be written as an oracle in depth $5$ as in Figure~\ref{fig:1}
where $\fbox2$ stands for the parity gate. Note how the fanouts are
indeed parallel. Now, any function $f$ over an arbitrary number
of input variables could be implemented with an oracle of the same
shape, of depth 5. It is also easy to see how to extend this technique
to the case of a boolean function $f$ with more than one output.

\bigskip The remainder of this paper is concerned with the
implementation of this decomposition in MBQC, or
equivalently~\cite{browne10} in a model of quantum circuit where
unbounded fan-outs, Hadamard and rotations around the $Z$-axis are
elementary gates.

\section{A useful equality}

The main problem is the use of multi-controlled NOTs. In order to
proceed with their decompositions using $R(\theta)$-gates, we
generalize the formula of Selinger relating
conjunction of 3 boolean variables with XOR \cite[Eq~(5)]{selinger} to
Equation~\eqref{eq:andxor}, relating the conjunction of $n$ boolean
variables with XOR.

As it is customary, we assimilate the boolean {\em false} with $0$ and
the boolean {\em true} with $1$. The conjunction is simply the
product, and we can transparently write boolean equations as equations
over integers. With these conventions, one can show how to compute
the conjunction of $n$ booleans using XORs. This amounts to the
Fourier spectra of the conjunction.

\begin{lem}
  \label{lem:3}
  For all $n>0$ and for any family $\{x_i\}_{i=1,\ldots,n}$ of
  booleans, and if $\pset{n}{i}$ is the set of all subsets of
  $\{1\ldots n\}$ of size equal to $i$,
  \begin{equation}
    \label{eq:andxor}
    2^{n-1}\bigwedge_{i=1}^nx_i
    =
    \sum_{i=1}^n (-1)^{i-1} \sum_{K\in\pset{n}{i}}\bigoplus_{k\in K}x_k
  \end{equation}
\end{lem}

\begin{proof}
  The proof is done by induction on $n$.
  
  For $n=1$, the equality is trivial.
  
  For $n=2$, the equality is
  \begin{equation}
    \label{eq:1}
    2x_1x_2 = x_1 + x_2 - x_1\oplus{}x_2
  \end{equation}
  which can be shown correct by inspection of the 4 possible values
  for the pair $(x_1,x_2)$.
  
  Now suppose that the equation is correct for $n\geq2$, and consider
  the case $n+1$:
  \[
  2^n\bigwedge_{i=1}^{n+1}x_i = 
  2x_{n+1}\cdot2^{n-1}\bigwedge_{i=1}^{n}x_i
  \]
  which is, by induction hypothesis, equal to
  \[
  2x_{n+1}\cdot\left(
    \sum_{i=1}^n (-1)^{i-1} \sum_{K\in\pset{n}{i}}\bigoplus_{k\in K}x_k
  \right).
  \]
  Expanding, this is equal to
  \[
  \sum_{i=1}^n (-1)^{i-1} \sum_{K\in\pset{n}{i}}2x_{n+1}\bigoplus_{k\in K}x_k.
  \]
  Using Eq.~\eqref{eq:1}, we get
  \[
  \sum_{i=1}^n (-1)^{i-1} \sum_{K\in\pset{n}{i}}\left(x_{n+1} + \bigoplus_{k\in K}x_k -
    x_{n+1}\oplus\bigoplus_{k\in K}x_k\right).
  \]
  One can then conclude using Lemma~\ref{lem:2} (found in the appendix).
\end{proof}

\section{Multi-controlled NOT gates}
\label{sec:ccnot}

Extending the technique presented in \cite{selinger}, together with
auxiliary ancillas one can decompose any multi-controlled $Z$-gate
as a circuit consisting of Clifford and $R(\theta)$ gates.

\begin{figure*}[tb]
  \[
  \begin{array}{c}\xymatrix@=2ex{
      \ar@{-}[r]&*={\bullet}\ar@{-}[r]\ar@{-}[d]&
      \\
      \ar@{-}[r]&*={\bullet}\ar@{-}[r]\ar@{-}[d]&
      \\
      \ar@{-}[r]&*+[F]{Z}\ar@{-}[r]&
    }
  \end{array}
  =
  \begin{array}{c}
    \includegraphics{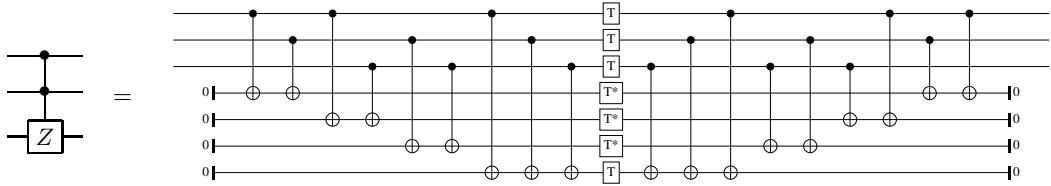}
  \end{array}
  \]
  \caption{Decomposition of the 2-controlled $Z$-gate.}
  \label{fig:toffoli}
\end{figure*}

\begin{figure*}[tb]
  \[
  \begin{array}{c}\xymatrix@=2ex{
      \ar@{-}[r]&*={\bullet}\ar@{-}[r]\ar@{-}[d]&
      \\
      \ar@{-}[r]&*={\bullet}\ar@{-}[r]\ar@{-}[d]&
      \\
      \ar@{-}[r]&*={\bullet}\ar@{-}[r]\ar@{-}[d]&
      \\
      \ar@{-}[r]&*+[F]{Z}\ar@{-}[r]&
    }
  \end{array}
  =
  \begin{array}{c}
    \includegraphics[width=6.1in]{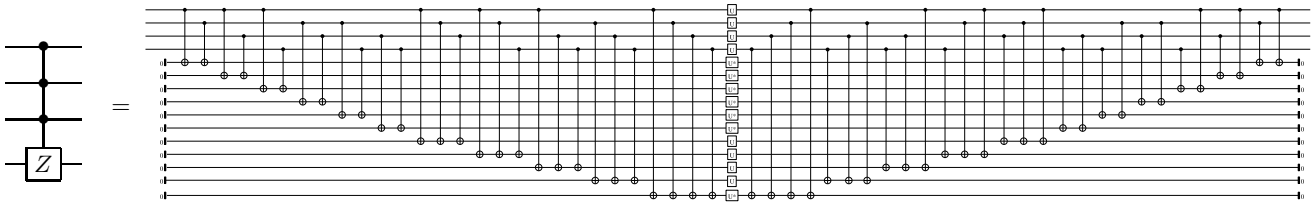}
  \end{array}
  \]
  \caption{Decomposition of the 3-controlled $Z$-gate.}
  \label{fig:toffoli3}
\end{figure*}

\begin{lem}
  \label{lem:4}
  Any Z-gate controlled by $n\geq 2$ qubits can
  be written as a circuit consisting of (1) a sequence of CNOTs, (2) a
  list of $2^{n+1}-1$ gates $R(\theta)$ in parallel, (3) a sequence of CNOTs.
\end{lem}

\begin{proof}
  The proof is an adaptation of the one developed by Selinger~\cite{selinger},
  generalized to the $n$-ary case. Let $T^n$ be the gate sending
  \[
  \ket{x_1\ldots{}x_n} \longmapsto
  (-1)^{x_1\wedge\cdots\wedge{}x_n}\ket{x_1\ldots{}x_n}
  \]
  with $n\geq 3$.
  This gate is a Z-gate controlled by $n-1$ quantum bits. 
  Thanks to Lemma~\ref{lem:3}, 
  $(-1)^{x_1\wedge\cdots\wedge{}x_n}$ can be written as
  \[
  \prod_{i=1}^{n}
  \prod_{K\in\pset{n}{i}}\omega_{n}^{(-1)^{i-1}\oplus_{k\in K}x_k}
  \]
  where $\omega_n=e^{\frac{i\pi}{2^n}}$. Therefore, the gate $T^n$ can
  be implemented by applying $R(\frac{i\pi}{2^{n-1}})$-gates and
  $R(\frac{-i\pi}{2^{n-1}})$-gates to qubits
  in state $\ket{\oplus_{k\in K}x_k}$ where $K$ are non-empty subsets
  of $\{1\ldots n\}$. One can construct and store these values using
  CNOT gates and $2^{n}-1-n$ ancillas: this allows the rotations gates to
  be set in parallel. The ancillas can then be reset to their original
  values, since the rotations around the $Z$-axis only change the
  phase.
\end{proof}

As examples, we first show the 2-controlled $Z$-gate \cite{selinger} in
Figure~\ref{fig:toffoli}: the $T$ gate is $R(\frac\pi4)$. We also show
the case of the $Z$-gate controlled by 3 qubits in
Figure~\ref{fig:toffoli3}, where the gate $U$ is $R(\frac\pi8)$.
In both cases, the blocks of CNOTS are indeed made of pairwise
commuting gates.

\bigskip
It is easy to see that in a given decomposition, each of the two
blocks of CNOTs can be made of pairwise commuting gates: each block
can then be encoded in constant-depth using unbounded fan-outs and
parity gates~\cite{green02,moore02}.

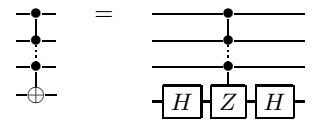
\begin{figure}[bt]
\[
\xymatrix@=1ex{
  \ar@{-}[r]&*{\bullet}\ar@{-}[r]\ar@{-}[d]&\\
  \ar@{-}[r]&*{\bullet}\ar@{-}[r]\ar@{..}[d]&\\
  \ar@{-}[r]&*{\bullet}\ar@{-}[r]\ar@{-}[d]&\\
  \ar@{-}[r]&*{\oplus}\ar@{-}[r]&
}
~~=~~
\xymatrix@=1ex{
  \ar@{-}[rr]&&*{\bullet}\ar@{-}[rr]\ar@{-}[d]&&\\
  \ar@{-}[rr]&&*{\bullet}\ar@{-}[rr]\ar@{..}[d]&&\\
  \ar@{-}[rr]&&*{\bullet}\ar@{-}[rr]\ar@{-}[d]&&\\
  \ar@{-}[r]&*+[F-]{H}\ar@{-}[r]&*+[F-]{Z}\ar@{-}[r]&*+[F-]{H}\ar@{-}[r]&
},
\]
\caption{Controlled NOTs and controlled $Z$-gates.}
\label{fig:NOTZ}
\end{figure}

Therefore, because multi-controlled NOTs are two Hadamard away from
multi-controlled $Z$-gates as shown in Figure~\ref{fig:NOTZ}, 
any multi-controlled NOT gate can be written in constant
depth using arbitrary $Z$-rotations, Hadamard gates, unbounded fanouts
and parity gates.

\section{Quantum oracles in MBQC}

Together with unbounded fanouts, Hadamard gates and arbitrary
rotations along the Z axis and using the technique presented in
Section~\ref{sec:naive}, one can therefore implement any boolean
function in constant depth. Since constant-depth circuits using such
gates can be implemented by constant-depth MBQC
patterns~\cite{browne10}, one concludes that quantum oracles can be
implemented in constant depth in measurement-based quantum
computation.

\section{Complexity of the overall size}

If the depth of the computation is constant, it is worth noting that
in general the overall size is exponential with respect to the size of
the input. Indeed, the width of the corresponing circuit corresponds
to the sum of the numbers of subsets of the $K_i$ in
Eq.~\eqref{eq:f}. For example, consider the function $f$ as the
conjunction, sending the vector $(x_1,\ldots,x_n)$ to
$x_1\wedge\ldots\wedge x_n$. This is computed by a NOT-gate controlled
by $n$ quantum bits, which, from Lemma~\ref{lem:4}, can be represented
by a quantum circuit consisting of $2^{n+1}-1$ $Z$-rotations. The
resulting MBQC pattern is therefore exponential in $n$.

\bigskip
One can however recover a polynomial sized-pattern for
Eq.~\eqref{eq:f} in the case where $N$ is polynomial in $n$ and when
the size of the $K_i$ is at most logarithmic in $n$. For example,
the generalization of Eq.~\eqref{eq:f4}
\[
f(x_1,\ldots x_n) = \bigoplus_{i\neq j} x_i\wedge x_j
\]
has a pattern representation of size polynomial on $n$.

\section{Disjunction in constant depth.}
We conclude this paper with a side comment, answering an open question.
H{\o}yer and \v{S}palek have asked~\cite{hoyer05} whether
the disjunction:
\[
\ket{x_1,\ldots,x_n}\ket{y}\longmapsto
\ket{x_1,\ldots,x_n}\ket{y\oplus(x_1\vee\ldots\vee x_n)}
\]
can be implemented exactly by a constant-depth circuit. Using the
results of the present paper, we can answer positively: using the fact
that the disjunction of $n$ variables $x_1\vee\ldots\vee x_n$ can be
realized with a simple conjunction ${\tt not}({\tt
  not}~x_1\wedge\ldots\wedge {\tt not}~x_n)$, the requested circuit is
essentially the decomposition of the multi-controlled NOT
gate. However, note that the {\em size} of the circuit is exponential
on $n$.

\section{Acknowledgments}

We would like to thank Simon Perdrix for enlightening
discussions. This work was supported by the ANR project
ANR-2010-BLAN-021301 LOGOI.

\appendix

\section{Auxiliary lemmas}

In this appendix we recall two elementary results about binomial coefficients.

Let us write $\pset{n}{i}$ for the set of all subsets of $\{1\ldots
n\}$ of size equal to $i$. If $X$ is a set, let us write $\sharp{}X$
for the size of $X$. Note that $\sharp\pset{n}{i}$ is the binomial coefficient
$\binom{n}{p}$.
\begin{lem}
  \label{lem:1}
  For all $n>0$, for all $0<i\leq n+1$, the following equality holds:
  $\binom{n+1}{i} = \binom{n}{i-1} + \binom{n}{i}$.
\end{lem}

\begin{proof}
  This is an easy corollary of the fact that the set $\pset{n+1}{i}$
  is in fact $\{S\cup\{n+1\}\,|\, S\in\pset{n}{i-1}\} \cup \pset{n}{i}$.
\end{proof}

\begin{lem}
  \label{lem:2}
  For all $n>0$,
  $
  \sum_{i=1}^n(-1)^{i-1}\binom{n}{i} = 1.
  $
\end{lem}

\begin{proof}
  If $n=1$, the lemma is true since there is only one element in a
  singleton. If $n>1$, then using the previous lemma we deduce that
  \[
  \sum_{i=1}^{n}(-1)^{i-1}\binom{n}{i}
  =
  \sum_{i=1}^{n}(-1)^{i-1}(\binom{n-1}{i-1} + \binom{n-1}{i}).
  \]
  This is equal to
  \[
  \binom{n-1}{0} + (-1)^{n-1}\binom{n-1}{n}.
  \]
  The first element in the sum is $1$, the second is $0$.
\end{proof}

\bibliography{paper}

\begin{thebibliography}{7}
\providecommand{\natexlab}[1]{#1}
\providecommand{\url}[1]{\texttt{#1}}
\expandafter\ifx\csname urlstyle\endcsname\relax
  \providecommand{\doi}[1]{doi: #1}\else
  \providecommand{\doi}{doi: \begingroup \urlstyle{rm}\Url}\fi

\bibitem[Broadbent and Kashefi(2009)]{broadbent09}
A.~Broadbent and E.~Kashefi.
\newblock Parallelizing quantum circuits.
\newblock \emph{Journal of Theoretical Computer Science}, 410\penalty0 (26),
  2009.

\bibitem[Browne et~al.(2010)Browne, Kashefi, and Perdrix]{browne10}
Dan Browne, Elham Kashefi, and Simon Perdrix.
\newblock Computational depth complexity of measurement-based quantum
  computation.
\newblock In \emph{Proceeding of the Fifth Conference on the Theory of Quantum
  Computation, Communication and Cryptography (TQC 2010)}, 2010.

\bibitem[Green et~al.(2002)Green, Homer, Moore, and Pollett]{green02}
F.~Green, S.~Homer, C.~Moore, and C.~Pollett.
\newblock Counting, fanout and the complexity of quantum {ACC}.
\newblock \emph{Quantum Information adn Computation}, 2\penalty0 (1):\penalty0
  35--65, 2002.

\bibitem[H{\o}yer and \v{S}palek(2005)]{hoyer05}
Peter H{\o}yer and Robert \v{S}palek.
\newblock Quantum fan-out is powerful.
\newblock \emph{Theory of Computing}, 1:\penalty0 81--103, 2005.

\bibitem[Moore and Nilsson(2002)]{moore02}
C.~Moore and M.~Nilsson.
\newblock Parallel quantum computation and quantum codes.
\newblock \emph{SIAM Journal on Computing}, 31\penalty0 (3):\penalty0 799--815,
  2002.

\bibitem[Raussendorf and Briegel(2001)]{RB01}
R.~Raussendorf and H.~J. Briegel.
\newblock Quantum computing via measurements only.
\newblock \emph{Physical Review Letters}, 86:\penalty0 5188--5191, 2001.

\bibitem[Selinger(2013)]{selinger}
Peter Selinger.
\newblock Quantum circuits of {T}-depth one.
\newblock \emph{Physical Review A}, 87:\penalty0 042302, 2013.

\end{thebibliography}
\bibliographystyle{plainnat}

\end{document}